\documentclass[a4paper, runningheads]{llncs}
\usepackage[utf8]{inputenc}

\usepackage{thmtools}
\usepackage{todonotes}
\usepackage{graphicx}
\usepackage{complexity}
\usepackage{wrapfig}
\usepackage{subfig}
\usepackage{xspace}
\usepackage{hyperref,cite}
\usepackage{amsmath, amssymb, ntheorem}
\usepackage{xpatch}
\usepackage{doi,url}

\hypersetup{
hidelinks,
colorlinks=true,
citecolor=[rgb]{0.121 0.47 0.705},
linkcolor=[rgb]{0.121 0.47 0.705},
urlcolor=[rgb]{0.121 0.47 0.705}
}

\title{Maximizing Ink in Partial Edge Drawings\\ of $k$-plane Graphs\thanks{The authors thank Michael H\"oller and Birgit Schreiber for the fruitful discussions during the ``Seminar in Algorithms: Graphs and Geometry'' held in 2017 at TU Wien. A preliminary abstract of this paper has been presented at EuroCG 2018.}}
\author{Matthias Hummel 
\and Fabian Klute 
\and Soeren Nickel 
\and Martin~N\"ollenburg} 
\institute{Algorithms and Complexity Group, TU Wien, Vienna, Austria \email{matthiashummel@ymail.com,[fklute|noellenburg]@ac.tuwien.ac.at, soeren.nickel@tuwien.ac.at}
}

\newcommand{\ped}{\ensuremath{\textsc{PED}}\xspace}
\newcommand{\sped}{\ensuremath{\textsc{SPED}}\xspace}
\newcommand{\maxsped}{\ensuremath{\textsc{MaxSPED}}\xspace}
\newcommand{\maxped}{\ensuremath{\textsc{MaxPED}}\xspace}

\newcommand{\ppsat}{\ensuremath{\textsc{planar 3-Sat}}\xspace}
\newcommand{\sollong}{\ensuremath{\textit{long}}\xspace}

\newcommand{\solshort}{\ensuremath{\textit{short}}\xspace}

\let\doendproof\endproof
\renewcommand\endproof{\hfill \qed \doendproof}

\begin{document}
	\maketitle
	
	\begin{abstract}
		Partial edge drawing (PED) is a drawing style for non-planar graphs, in which edges are drawn only partially as pairs of opposing stubs on the respective end-vertices. 
		In a PED, by erasing the central parts of edges, all edge crossings and the resulting visual clutter are hidden in the undrawn parts of the edges.
		In symmetric partial edge drawings (SPEDs), the two stubs of each edge are required to have the same length. 
		It is known that maximizing the ink (or the total stub length) when transforming a straight-line graph drawing with crossings into a SPED is tractable for 2-plane input drawings, but \NP-hard for unrestricted inputs.
		We show that the problem remains \NP-hard even for 3-plane input drawings and establish \NP-hardness of ink maximization for PEDs of 4-plane graphs.
		Yet, for $k$-plane input drawings whose edge intersection graph forms a collection of trees or, more generally, whose intersection graph has bounded treewidth, we present efficient algorithms for computing maximum-ink PEDs and SPEDs.
		We implemented the treewidth-based algorithms and show a brief experimental  evaluation.
	\end{abstract}

	\section{Introduction}
	
	Visualizing non-planar graphs as node-link diagrams is challenging due to the visual clutter caused by edge crossings. The layout readability deteriorates as the edge density and thus the number of crossings increases~\cite{p-wagehu-97}.
	Therefore alternative layout styles are necessary for non-planar graphs.
	A radical approach first used in applied network visualization work by Becker et al.~\cite{bew-vnd-95} is to start with a traditional straight-line graph drawing and simply drop a large central part of each edge and with it many of the edge crossings.
	This idea relies on the closure and continuation principles in Gestalt psychology~\cite{k-pgp-35}, which imply that humans can still see a full line segment based only on the remaining edge stubs by filling in the missing information.
	User studies have confirmed that such drawings remain readable while reducing clutter significantly~\cite{bvkw-epdldge-12,bkl-us-15} and Burch et al.~\cite{Burch2014} presented an interactive graph visualization tool using partially drawn edges combined with fully drawn edges.
	
	\begin{figure}
		\centering
		\subfloat{
			\includegraphics{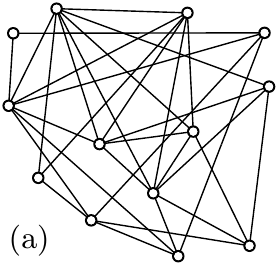}
		}
		\subfloat{
			\includegraphics{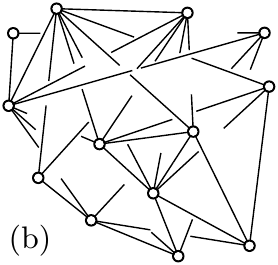}
		}
		\subfloat{
			\includegraphics{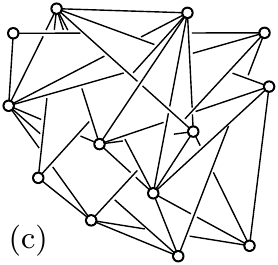}
		}
		\subfloat{
			\includegraphics{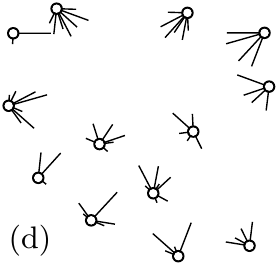}
		}
		\caption{Drawings of the same graph. (a) A straight-line drawing, (b) a maximum-ink SPED, (c) a maximum-ink PED, and (d) a maximum-ink SHPED.}\vspace{-.5cm}
		\label{fig:examples}
	\end{figure}

	The idea of drawing edges only partially has subsequently been formalized in graph drawing as follows~\cite{bk-ecbe-12}. 
	A \emph{partial edge drawing (PED)} is a graph drawing that maps vertices to points and edges to pairs of crossing-free edge stubs of positive length pointing towards each other.
	These edge stubs are obtained by erasing one contiguous central piece of the straight-line segment connecting the two endpoints of each edge.
	In other words each straight-line edge is divided into three parts, of which only the two outer ones are drawn (see Fig.~\ref{fig:examples}).
	More restricted and better readable~\cite{blmt-pedhmitc-16} variations of PEDs are \emph{symmetric} PEDs, in which both stubs of an edge must have the same length (see Fig.~\ref{fig:examples}(b)), and \emph{homogeneous} PEDs, in which the ratio of the stub length to the total edge length is the same constant for all edges.
	Symmetric stubs facilitate finding adjacent vertices due to the identical stub lengths at both vertices, and symmetric homogeneous stubs (see Fig.~\ref{fig:examples}(d)) additionally indicate the distance at which to find a neighboring vertex.
	Clearly, for very short stubs it is easy to hide all edge crossings, but reading adjacency information gets very difficult~\cite{bvkw-epdldge-12}.
	Therefore, the natural optimization problem in this formal setting is \emph{ink maximization}, i.e., maximizing the total stub length, so that as much information as possible is given in the drawing while all crossings disappear in the negative background~space.

	We study the ink maximization problem for partial edge drawings (PEDs) and symmetric partial edge drawings (SPEDs) with a given geometric input drawing.
	These problems are known as \maxped and \maxsped, respectively~\cite{bcgkmn-pped-17,bk-ecbe-12}. 
	Note that with a given input drawing, the ink maximization problem for symmetric homogeneous PEDs (SHPEDs) is trivial, as we can simply iterate over all crossings, choose the larger of the two stub ratios resolving the crossing and take the minimum of all these stub ratios, which yields the best solution.

	\medskip
	
	\noindent\textbf{Related Work.} As a first result, Bruckdorfer and Kaufmann~\cite{bk-ecbe-12} presented an integer linear program for solving \maxsped on general input drawings. 
	Later, Bruckdorfer et al.~\cite{bcgkmn-pped-17} gave an $O(n \log n)$-time algorithm for \maxsped on the class of 2-plane input drawings (no edge has more than two crossings), where $n$ is the number of vertices.
	They also described an efficient 2-approximation algorithm for the dual problem of minimizing the amount of erased ink for arbitrary input drawings.
	The PhD thesis of Bruckdorfer~\cite{b-sgh-15} presents a sketch of an \NP-hardness proof for \maxsped, but left the complexity of \maxped as an open problem, as well as the design of algorithms for \maxped. 
	
	There are a number of additional results for PEDs without a given input drawing, i.e., having the additional freedom of placing the vertices in the plane. 
	For example, the existence or non-existence of SHPEDs with a specified stub ratio $\delta$ for certain graph classes such as complete graphs, complete bipartite graphs, or graphs of bounded bandwidth has been investigated~\cite{bk-ecbe-12,bcgkmn-pped-17}.
	From a practical perspective, Bruckdorfer et al.~\cite{bkl-pa1-15} presented a force-directed layout algorithm to compute SHPEDs for stubs of $1/4$ of the total edge length, but without a guarantee that all crossings are eliminated. Moreover, the idea of partial edge drawings has also been extended to orthogonal graph layouts~\cite{bkm-1oped-14}.
		
	\medskip
	
	\noindent\textbf{Contribution.} We extend the results of Bruckdorfer et al.~\cite{bcgkmn-pped-17} on 2-plane geometric graph drawings to $k$-plane graph drawings for $k > 2$, where a given graph drawing is \emph{$k$-plane}, if every edge has at most $k$ crossings.
	In particular, we strengthen the \NP-hardness of \maxsped~\cite{b-sgh-15} to the case of 3-plane input drawings without three (or more) mutually crossing edges.
	 For \maxped we show \NP-hardness, even for 4-plane input drawings, which settles a conjecture of Bruckdorfer~\cite{b-sgh-15}.
	On the positive side, we give polynomial-time dynamic programming algorithms for both \maxsped and \maxped of $k$-plane graph drawings whose edge intersection graphs are collections of trees. 
	More generally, we extend the algorithmic idea and obtain FPT algorithms if the edge intersection graph has bounded treewidth 
	and also provide a proof-of-concept implementation.
	We evaluate the implementation using non-planar drawings from two classical layout algorithms, namely a force-based and a circular layout algorithm.

	\smallskip \noindent {\emph{Statements whose proofs are located in the appendix are marked with $\star$.}}

	\section{Preliminaries}
	\label{sec:preliminaries}

	Let $ G $ be a \emph{simple graph} with edge set $ E(G) = S = \{s_1,\dots,s_m\}$ and $ \Gamma $ a straight-line drawing of $ G $ in the plane. We call $ \Gamma $ \emph{$ k $-plane} if every edge $ s_i \in S $ is crossed by at most $ k $ other edges from $ S $ in $ \Gamma $. We often use 
	edge in $ S $ and segment in $ \Gamma $ interchangeably. %
	Hence $ S $ can be seen as a set of line segments.
	
	The \emph{intersection graph} $ C $ of $ \Gamma $ is the graph containing a vertex $v_i$ in $ V(C) $ for every $ s_i \in S $ and an edge $ v_i v_j \in E(C) $ between vertices $ v_i, v_j \in V(C) $ if the corresponding edges $ s_i, s_j \in S $ intersect in $ \Gamma $. 
	We also denote the segment in $S$ corresponding to a vertex $v \in V(C)$ by $s(v)$.
	Observe that the intersection graph $ C $ of a $ k $-plane drawing $ \Gamma $ has maximum degree~$ k $.  
	Using a standard sweep-line algorithm, computing the intersection graph $C$ of a set of $m$ line segments takes $O(m \log m + |E(C)|)$ time~\cite{bcko-cgaa-08}, where $|E(C)|$ is  the number of intersections.
	
	A \emph{partial edge drawing} (PED) $ D $ of $ \Gamma $ draws a fraction $ 0 < f_s \le 1 $ of each edge $ s = uv \in S $ by drawing edge stubs of length $ f_u |s| $ at $ u $ and $ f_v |s| $ at $ v $, s.t., $ f_u + f_v = f_s $. The \emph{ink} or \emph{ink value}  $I(D)$ of a PED $D$ is the total stub length $I(D) = \sum_{s \in S} f_s |s|$.
	In the problem \maxped, the task is to find for a given drawing~$\Gamma$ a PED $D^*$ such that $I(D^*)$ is maximum over all PEDs.
	A \emph{symmetric partial edge drawing} (SPED) $D$ of $\Gamma$ is a PED, s.t., $ f_u = f_v = f_s/2 $ for every edge $ s = uv \in S $. Then the \maxsped problem is defined analogously to \maxped.

	\paragraph{Treewidth.}
A \emph{tree decomposition}~\cite{robertson1984graph} for a graph $ G $ is a pair $ (T,\mathcal{X}) $ with $ T $ being a tree and $ \mathcal{X} $ a collection of subsets $ X_i \subseteq  V(G)$. For every edge $ uv \in E(G) $ we find $ t \in V(T) $ such that $ \{u,v\} \subseteq X_t $ and for every vertex $ v \in V(G) $ we get $ T[\{t \mid v \in X_t\}] $ is a connected and non-empty subtree of $ T $.
	To differentiate the vertices of $ G $  and $ T $ we call the vertices of $ T $ \emph{nodes} and a set $ X_i \in \mathcal{X} $ a \emph{bag}. Now the \emph{width} of a tree decomposition $ (T,\mathcal{X}) $ is defined as $ \max\{|X_t| - 1 \mid t \in V(T)\} $. For a graph $ G $ we say it has \emph{treewidth} $ \omega $, if the tree decomposition with minimum width has width $ \omega $. For a node $ t \in T $ we denote with $ V_t\subseteq V(G) $ the union of all bags $ X_{t'} \in \mathcal X $ such that $ t' $ is either $ t $ or a descendent of $ t $ in $ T $.

	In our algorithms we are using the well known \emph{nice tree decomposition}~\cite{downey2012parameterized}. For a graph $ G $ a nice tree decomposition $ (T,\mathcal{X}) $ is a special tree decomposition, where $ T $ is a rooted tree and we require that every node in $ T $ has at most two children. In case $ t \in V(T) $ has two children $ t_1, t_2 \in T $, then $ X_t = X_{t_1} = X_{t_2} $. Such a node is called \emph{join node}. For a node $ t \in T $ with a single child $ t_1\in T $ we find either $ |X_t| = |X_{t_1}| + 1 $, $ X_{t_1} \subset X_t $ or $ |X_t| = |X_{t_1}| - 1 $, $ X_t \subset X_{t_1}$. The former we call \emph{insert node} and the latter \emph{forget node}. A leaf $ t \in T $ is called a \emph{leaf node} and its bag contains a single vertex. Finally let $ r \in T $ be the root of $ T $, then $ X_r = \emptyset $. It is known that a tree decomposition can be transformed into a nice tree decomposition of the same width $ \omega $ and with $ O(\omega |V(G)|) $ tree nodes in time linear in the size of the graph $ G $~\cite{downey2012parameterized}.

	\section{Complexity}
	We first investigate the complexity of \maxsped and \maxped, and prove both problems to be \NP-hard for 3-plane and 4-plane input drawings, respectively. %

	\subsection{Hardness of \maxsped for $k\ge3$}
	\label{sec:hardness}

	We close the gap between the known \NP-hardness of \maxsped~\cite{b-sgh-15} for general input drawings and the polynomial-time algorithm for 2-plane drawings~\cite{bcgkmn-pped-17}.

	\begin{theorem}\label{thm:hard}
		\maxsped is \NP-hard even for 3-plane graph drawings.
	\end{theorem}

\noindent \textit{Proof.}
		We reduce from the \NP-hard problem \ppsat~\cite{l-pftu-82} using similar ideas as in Bruckdorfer's sketch of the hardness proof for general \maxsped~\cite{b-sgh-15}. 
		Here we specify precisely the maximum ink contributions of all gadgets needed for a satisfying variable assignment. 
		Our variable gadgets are cycles of edge pairs that admit exactly two maximum-ink states.
		We construct clause gadgets consisting of three pairwise intersecting edges so that all crossings are between two edges only, while Bruckdorfer's gadgets have multiple edges intersecting in the same point.
		Let $\phi$ be a planar \textsc{3-Sat} formula with $n$ variables $\{x_1, \dots, x_n\}$ and $m$ clauses $\{c_1, \dots, c_m\}$, each consisting of three literals.
		We can assume that $\phi$ comes with a planar drawing of its variable-clause graph $H_\phi$, which has a vertex for each variable $x_i$ and a vertex for each clause $c_j$.
		Each clause vertex is connected to the three variables appearing in the clause.
		In the drawing of $H_\phi$ all variable vertices are placed on a horizontal line and the clause vertices connect to the adjacent variable vertices either from above or from below the horizontal line.
		In our reduction (see Fig.~\ref{fig:gadget_connection}) we mimic the drawing of $H_\phi$ by creating a 3-plane drawing $\Gamma_\phi$ as a set of line segments of two distinct lengths and define a value $L$ such that $\Gamma_\phi$ has a SPED with ink at least $L$ if and only if $\phi$ is satisfiable.
		The whole construction will be drawn onto a triangular grid of polynomial size.

		\begin{wrapfigure}[14]{r}{.6\textwidth}
			\centering
			\vspace{-1cm}
			\includegraphics{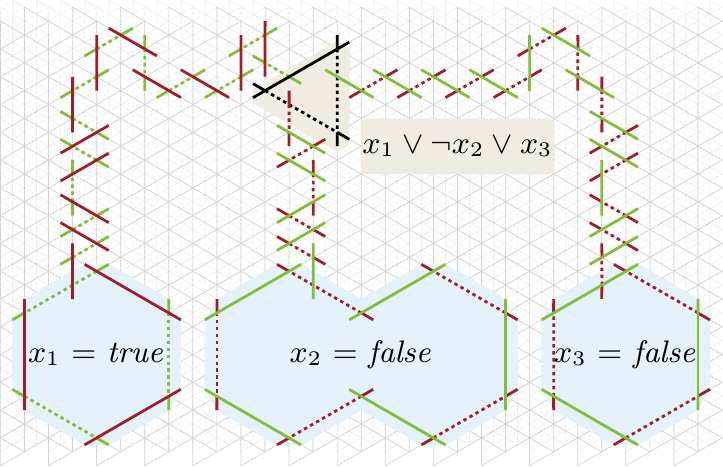}
			\caption{%
				Three variables and a satisfied clause gadget. Dotted parts do not belong to the SPED.}
			\label{fig:gadget_connection}
		\end{wrapfigure}
		
		All segments in the clause or variable gadgets are of length 8.
		The segments used for the connections are of length 4.
		We use pairs of intersecting segments, alternatingly colored red and green.
		The intersection point of each red-green segment pair is at distance 1 from an endpoint.
		Thus, the maximum amount of ink contributed by such a pair is 10 or 6, respectively (one full segment of length 8 or 4, respectively, and two stubs of length 1 each).
		
		Each variable gadget is a cycle of segment pairs, with (at least) one pair for each occurrence of the variable in $\phi$, see Fig.~\ref{fig:gadget_connection}. 
		This cycle has exactly two ink-maximal SPEDs: either all red edges are full segments and all green edges are length-1 stubs or vice versa.
		We associate the configuration with green stubs and full red segments with the value \emph{true} and the configuration with red stubs and full green segments with \emph{false}.
		
		For each clause we construct a triple of mutually intersecting segments, see the gadget on yellow background in the upper part of Fig.~\ref{fig:gadget_connection}.
		Again, their intersection points are at distance~1 from the endpoints.
		It is clear that in such a clause triangle at most one of the three segments can be fully drawn, while the stubs of the other two can have length at most~1.
		Hence, the maximum amount of ink in a SPED contributed by a clause gadget is~12.
		
		Finally, we connect variable and clause gadgets, such that a clause gadget can contribute its maximum ink value of 12 if and only if the clause is satisfied by the selected truth assignment to the variables.
		For a positive (negative) literal, we create a path of even length between a green (red) edge of the variable gadget and one of the three edges of the clause gadget as shown in Fig.~\ref{fig:gadget_connection}.
		The first edge~$s$ of this path intersects the corresponding variable segment $s'$ such that $s'$ is split into a piece of length 3 and a piece of length 5, whereas $s$ is split into a piece of length 1 and a piece of length 3.
		The last edge of the path intersects the corresponding clause edge again with a length ratio of 3 to 5.
		The path consists of a chain of red-green segment pairs, each  contributing an ink value of at most~6.
		
		It remains to argue that the resulting drawing has polynomial size and is a correct reduction. 
		All segments are drawn on the underlying triangular grid and have integer lengths; all intersection points are grid points, too. 
		Since the drawing of $H_\phi$ has polynomial size, so do the constructed gadgets.
		Additionally, no segment intersects more than three other segments, so the  drawing is 3-plane. 
		
		For the correctness of the reduction, let $L$ be the ink value obtained by counting 10 for each red-green segment pair in a variable, 6 for each red-green segment pair in a wire, and 12 for each clause gadget.
		First assume that $\phi$ has a satisfying truth assignment and put each variable gadget in its corresponding state.
		For each clause, select exactly one literal with value \emph{true} in the satisfying truth assignment. 
		We draw the clause segment that connects to the selected literal as a full segment and the other two as length-1 stubs.
		Recall that the literal paths are oriented from the variable gadget to the clause gadget.
		Since the last segment of the selected literal path must be drawn as two length-1 stubs, the only way of having a maximum contribution of that path is by alternating stubs and full segments.
		Hence, the first segment of the path must be a full segment.
		But because the variable is in the state that sets the literal to \emph{true}, the intersecting variable segment is drawn as two stubs and the path configuration is valid.
		For the two non-selected literals, we can draw the last segments of their paths as full segments, as well as every segment at an even position, while the segments at odd positions are drawn as stubs. 
		This is compatible with any of the two variable configurations and proves that we can indeed achieve ink value~$L$.
		
		Conversely, assume that we have a SPED with ink value $L$. 
		By construction, every red-green segment pair and every clause gadget must contribute its respective maximum ink value.
		In particular, each variable gadget is either in state \emph{true} or \emph{false}.
		By design of the gadgets it is straight-forward to verify that the corresponding truth assignment satisfies~$\phi$. \qed

 	\subsection{Hardness of \maxped for $k\ge4$}
		\label{sec:hardness_maxped}
		We adapt our \NP-hardness proof for \maxsped to show  that \maxped is \NP-hard for $k$-plane drawings with $ k \geq 4 $.

		\begin{restatable}[$ \star $]{theorem}{thmmaxpedhard}\label{thm:maxped_hard}
			\maxped is \NP-hard even for 4-plane graph drawings.
		\end{restatable}

		\noindent \textit{Proof (Sketch).} As in the proof of Theorem~\ref{thm:hard}, we show the result via a reduction from \textsc{planar 3-Sat}. The key change for \maxped comes from the fact, that the two stubs are now independent from each other. Take two crossing edges as an example. We now can draw the two segments with almost full ink value by just excluding an $ \varepsilon $-sized gap in one of the two segments for some small $ \varepsilon > 0$. We will use this placement of a gap in the variable and wire gadgets, to create two possible states. As before we use an underlying triangular grid, which we omit in the figures of this section for ease of presentation.

		Let $ \phi $ be a planar \textsc{3-Sat} formula. For one variable $ x $ of $ \phi $ we construct a variable gadget consisting of a cycle of $ p $ line segments $ t_1,\dots,t_p $, see Fig.~\ref{fig:variable_sketch}. Such a cycle has exactly two maximum-ink drawings. One, where for each segment $ t_i $ the gap is placed at its intersection with $ t_{i+1} \pmod{p}$ and another drawing, in which the gap is placed at its intersection with $ t_{i-1} \pmod{p}$.
		 Figure~\ref{fig:variable_sketch} shows a gadget in its true state. 
		We set the length of one segment $ t_{i} $ to $ \alpha + 2\beta $, where $\alpha \in \mathbb N$ is the distance between the two intersection points and $\beta$ is the length of each stub sticking out.
		
		A clause gadget is a cycle of three pairwise intersecting segments $ r_1, r_2, r_3$, which we call triangle segments. All segments are elongated at one end, such that the total length of a segment $ r_i $, $ i \in \{1,2,3\} $, is $4\alpha + 2\beta$. Since the stubs are independent we could draw all three triangle segments. To avoid this we attach a big 4-cycle to each $r_i$. Then $ r_i $ intersects the 4-cycle at a segment $r_w$, see Fig.~\ref{fig:clause_sketch}. If we place the gap of $r_w$ at its intersection with $r_i$, we lose more units of ink than we gain by drawing every triangle segment $ r_i $ completely. Hence it is never possible to draw more than one full triangle segment in an ink-maximal PED.

		Finally, a wire is a chain of segments $ s_1,\dots,s_z $ for each variable occurrence in a clause $ c $ in $ \phi $. We place the wire such that $ s_1 $ intersects the corresponding variable gadget at some segment $ t_j $, and $ s_z $ intersects the clause gadget of $ c $ at one if its triangle segments $ r_i $. For the variable we place this intersection point at distance $ \beta $  to its intersection with $ s_{i+1} $, if it occurs positively, or with $ s_{i - 1} $, if it occurs negated. At the clause gadget we place the intersection of $ s_z $ with the corresponding $ r_i $ at distance $ \beta $ from the intersection between $ r_i $ and its successor $ r_{i+1} $, see the small squares in Fig.~\ref{fig:clause_sketch}. Each segment $ s_i $ has length $ \alpha/2 + 2\beta $.

		\begin{wrapfigure}[13]{r}{.62\textwidth}
			\centering
			\vspace{-1.15cm}
			\subfloat[Variable gadget]{
				\centering
				\includegraphics{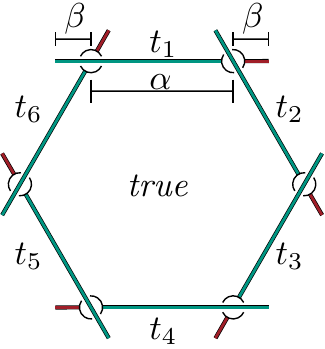}
				\label{fig:variable_sketch}
			}
			\subfloat[Clause gadget]{
				\centering
				\includegraphics{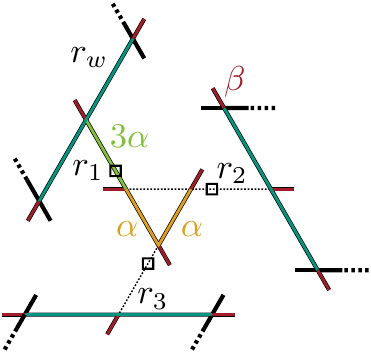}
				\label{fig:clause_sketch}				
			}
			\caption{Gadgets of our reduction. 
				Squares mark connection points for wires.}
		\end{wrapfigure}			

		Correctness follows  similarly to the proof of Theorem~\ref{thm:hard}. 
		Let $ \Gamma_\phi $ be the set of line segments constructed as above for a planar \textsc{3-Sat} formula $ \phi $. %
		We determine an ink value $ L $, s.t., $ \Gamma_\phi $ has a PED $ D $ with $ I(D) \geq L $ if and only if $ \phi $ has a satisfying variable assignment. The key property is that for each clause we find one wire such that its last segment is forced to place its gap at the intersection with the clause gadget in an ink-maximal PED. Then for each other segment $ s_i$, $ i \in \{1,\dots,z-1\} $, we must place its gap at the intersection with $ s_{i+1} $. Otherwise we would have to remove either $ \alpha/2 $ units of ink in the middle part of some $ s_i $, or remove $ \alpha -\beta $ units of ink at the variable gadget intersected by $ s_1 $. We show that both can be avoided if and only if $ \phi $ has a satisfying assignment. \hfill \qed %

 	\begin{corollary}\label{cor:nofptmaxped}
		\maxped and \maxsped for $k$-plane drawings are not fixed-parameter tractable, when parameterized solely by $k$.
	\end{corollary}

	\section{Algorithms}
	Sections~\ref{sec:hardness} and~\ref{sec:hardness_maxped} showed that \maxsped and \maxped are generally \NP-hard for $k \ge 3$ and $ k \geq 4 $ respectively. 
	Now we consider the special case that the intersection graph of the $k$-plane input drawing is a tree or has bounded treewidth. 
		In both cases we present polynomial-time dynamic programming algorithms for \maxsped (Sections~\ref{sec:tree} and~\ref{sec:btreewidth}) and \maxped (Section~\ref{sec:alg:maxped}).

	Let $ C $ be the intersection graph of a given drawing $\Gamma$ of a graph $ G $ as defined in Section~\ref{sec:preliminaries}. 
	Let $ u \in V(C) $ and $ \delta_u  = \deg(u) $. 
	Then for the corresponding segment $ s(u) \in S $ there are $ \delta_u + 1 $ relevant stub pairs including the whole segment, see Fig.~\ref{fig:tree}. 
	Let $ \ell_1(u), \dots, \ell_{\delta_u}(u) \in \mathbb{R}_+$ be the stub lengths induced by the intersection points of $s(u)$ with the segments of the neighbors of $u$,  sorted from shorter to longer stubs. We define  $ \ell_0(u) $ as the length of the whole segment $ s(u) $. 

	\begin{figure}[tbp]
		\centering
		\includegraphics[width=\textwidth]{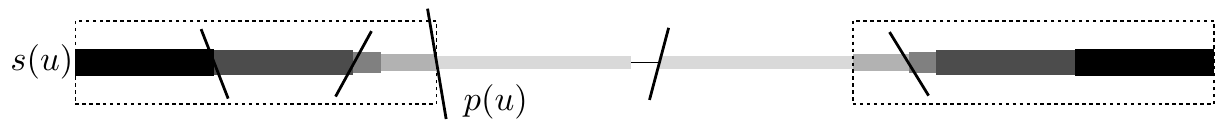}
		\caption{A segment $s(u)$ with five intersecting segments and the induced stub lengths. The boxed stub lengths are considered in $\solshort(u)$ and do not affect~$p(u)$.} 
		\label{fig:tree}
	\end{figure}

	\subsection{Trees}
	\label{sec:tree}

	Here we assume that $ C $ is a rooted tree of maximum degree $ k $. 
	We give a bottom-up dynamic programming algorithm for solving \maxsped on $ C $. 
	For each vertex $u \in V(C) $ we compute and store the maximum ink values $W_i(u)$ for $i = 0, \dots, \delta_u$ for the subtree rooted at $u$ such that $s(u)$ is drawn as a pair of stubs of length $\ell_i(u)$.
	For $ u \in V(C) $ let $ p(u) $ denote the parent of $ u $ in $ C $ and let $ c(u) $ denote the set of its children. 
	For $ u \in V(C) $  %
	let $ i_p $ be the index of the stub length $\ell_{i_p}(u)$ induced by the intersection point of $s(u)$ and $s(p(u))$.
	We define the following two values, which allow us to categorize the stub lengths into those not affecting the stubs of the parent and those that do affect the parent:
	\begin{align*}
		\solshort(u) = \max\{W_1(u),\dots,W_{i_p}(u)\} &&
		\sollong(u) = \max\{W_0(u),\dots,W_{\delta_u}(u)\}.
	\end{align*}
	Figure~\ref{fig:tree} highlights the stub lengths that are considered in $\solshort(u)$. We recursively define %
	\begin{align}
		\label{rec:tree}
		W_i(u) = \ell_i(u) + \sum_{v \in c(u)}
		\begin{cases}
		\solshort(v) & \text{if } s(u) \text{ with length } \ell_i(u) \text{ intersects } s(v) \\
		\sollong(v) & \text{otherwise.}
		\end{cases}
	\end{align}
	The correctness of Recurrence~(\ref{rec:tree}) follows by induction. 
	For a leaf $u$ in $C$ the set $c(u)$ is empty and the correctness of $W_i(u)$ is immediate.
	Further, $\solshort(u) = W_1(u)$ and $\sollong(u) = W_0(u)$ are set correctly for the parent $p(u)$.
	For an inner vertex $ u $ with degree $\delta_u$ we can assume by the induction hypothesis that the values $ \solshort(v) $ and $ \sollong(v) $ are computed correctly for all children $ v \in c(u) $. 
	Each value $ W_i(u) $ for $ 0 \leq i \leq \delta_u $ is then the stub length $ \ell_i(u) $ plus the sum of the maximum ink we can achieve among the children subject to the stubs of $u$ being drawn with length $\ell_i(u)$. 
	Setting $ \sollong(u) $ and $ \solshort(u) $ as above yields the two maximum ink values that are relevant for $ p(u) $.

	Recurrence~(\ref{rec:tree}) can be solved naively in $ O(mk^2) $ time, where $m = |V(C)|$. 
	Using the order on the stub lengths  we can improve this to $ O(mk) $ time by computing all $ W_i(u) $ for one $ u \in V(C) $ in $O(k)$ time.
	Let $u \in V(C)$ be a vertex with degree $\delta_u$. 
	The values $W_0(u) = \ell_0(u) + \sum_{v\in c(u)} \solshort(v)$ and $W_1(u) = \ell_1(u) + \sum_{v\in c(u)} \sollong(v)$ for the whole segment $s(u)$ and the shortest stubs can be computed in $O(k)$ time each.
	Now $W_{j+1}(u)$ can be computed from $W_j(u)$ in $O(1)$ time as follows.
	Let $v_j $ be the neighbor of $u$ that induces stub length $\ell_j(u)$ and assume $v_j \ne p(u)$.
	In $W_j(u)$ we could still count the value $\sollong(v_j)$, but in $W_{j+1}(u)$ the stub length of $u$ implies that $v_j$ can contribute only to $\solshort(v_j)$. 
	Then $W_{j+1}(u) = W_j(u) - \sollong(v_j) + \solshort(v_j)$.
	If $v_j=p(u)$ the two values $W_j(u)$ and $W_{j+1}(u)$ are equal, as the corresponding change in stub length has no effect on the children of $u$.
	Then computing $\solshort(u)$ and $\sollong(u)$ takes $O(k)$ time.

	Using standard backtracking we are able to find an optimal solution to the \maxsped problem on $ G $ with drawing $ \Gamma $ by solving Recurrence~(\ref{rec:tree}) in $ O(mk) $ time. Furthermore, the intersection graph $ C $ with $ m $ edges can be computed in $ O(m\log m) $ time. We obtain the following theorem.

	\begin{theorem}
		\label{thm:tree}
		Let $ G $ be a simple graph with $m$ edges and $ \Gamma $ a straight-line drawing of $ G $. If the intersection graph $ C $ of $ \Gamma $ is a tree with maximum degree $ k \in \mathbb{N} $, then problem \maxsped can be solved in $ O(mk + m \log m) $ time and space.
	\end{theorem}

	\subsection{Bounded Treewidth}
	\label{sec:btreewidth}
	Now we extend the case of a simple tree to the case that the intersection graph $C$ %
	has treewidth at most $\omega$. Our algorithm and proof of correctness follow a similar approach as the weighted independent set algorithm presented by Cygan et al.~\cite{cygan2015parameterized}. Let $ (T, \mathcal X) $ be a nice tree decomposition of $ C $ and $ k $ the maximum degree in $ C $. For $ V' \subseteq V(C) $ we define the \emph{stub set} $ S(V') := \{(u,\ell_i(u)) \mid u \in V' \text{ and } i = 0,\dots,\delta_u\} $. For $ (u,\ell_u), (v,\ell_v) \in S(V') $, $ u\neq v $, we say $ (u,\ell_u) $ \emph{intersects} $ (v,\ell_v) $ if $ s(u) $ drawn with stub length $ \ell_u $ intersects $ s(v) $ drawn with length $ \ell_v $. Further we call $ S(V') $ \emph{valid} if $ S(V') $ contains exactly one pair $ (u,\ell) $ for each $ u \in V' $ and no two pairs in $ S(V') $ intersect, i.e., the stub lengths in $ S(V') $ imply a \sped in the input drawing $ \Gamma $. Further we define the ink of a stub set $S(V')$ as $ I(S(V')) := \sum_{(u,\ell) \in S(V')}\ell $.
		
	\begin{restatable}[$ \star $]{lemma}{lemintersection}
		\label{lem:intersection}
		Let $ G $ be a simple graph, $ \Gamma $ a straight line drawing of $ G $, $ C $ the intersection graph of $ \Gamma $, and $ (T, \mathcal X) $ a tree decomposition of $ C $. For any fixed $ S \subseteq S(X_t), t\in T $, any two valid stub sets $ S_1, S_2 \subseteq S(V(C)) $ with maximum ink and $ S_1 \cap S(X_t) = S_2 \cap S(X_t) = S $ have equal
		ink value $ I(S_1 \cap S(V_t)) = I(S_2 \cap S(V_t))$.
	\end{restatable}

	As a consequence of Lemma~\ref{lem:intersection} we get that it suffices to store for every set of vertices $ V_t $ and a node $ t \in T $ the value of a maximum-ink valid stub set for the choices of vertices in $ S(X_t) $. Let $ t \in T $ and $ S \subseteq S(X_t) $ a stub set, then we define
	\[W(t, S) = \max \{I(\hat{S}) \mid \hat{S} \text{ is a valid stub set, } S \subseteq \hat{S} \subseteq S(V_t) \text{, and } \hat{S} \cap S(X_t) = S \}.\]
	If no such $ \hat{S} $ exists, we set $ W(t,S) = -\infty $.	In other words, $ W(t,S) $ is the maximum ink value achievable by any valid stub set in $ S(V_t) $ choosing $ S $ as the stub set for the vertices in $ X_t $. If the values $ W(t,S) $ are computed correctly for every $ t \in T $ we find the ink-value of a maximum-ink \sped by evaluating $ W(r,\emptyset) $ with $ r $ being the root of $ T $. Applying standard backtracking we can also reconstruct the stubs themselves. We now describe how to compute $ W(t,S) $ for every node type $ t \in T $ of a nice tree decomposition of $ C $. All the recursion formulas are stated here, as well as the correctness proof for the introduce nodes, the correctness proofs for the forget and join nodes can be found in Appendix~\ref{app:tw}.
	
	\medskip
	\noindent\textit{Leaf node.} Let $ t \in T $ be a leaf node and $ v \in X_t $ the vertex contained in its bag, then we store $ W(t,\{(v,\ell_i(v))\}) $ for each pair $ (v,\ell_i(v)) \in S(X_t) $ with $ i \in [0,\delta_v] $.

	\medskip	
	\noindent\textit{Introduce node.} Suppose next $ t \in T $ is an introduce node and $ t' $ its only child.  Let $ v \in X_t$ be the vertex introduced by $ t $, $ S \subseteq S(X_{t}) $, and $ i  \in [0,\delta_v]$, s.t., $ (v,\ell_i(v)) \in S $. If $ S $ is not a valid stub set we set $ W(t,S) = -\infty $, else 
	\begin{align*}
		W(t,S) = W(t',S\setminus \{(v,\ell_i(v))\}) + \ell_i(v).
	\end{align*}
	Correctness follows by considering a valid stub set $ \hat{S} $ whose maximum is attained in the definition of $ W(t,S) $. Then for the set $ \hat{S} \setminus \{(v,\ell_i(v))\} $ it follows that it is considered in the definition of $ W(t',S\setminus\{(v,\ell_i(v))\}) $ and hence we get 
	\begin{align*}
		 W(t',S\setminus\{(v,\ell_i(v))\}) &\geq I(\hat{S}\setminus \{(v,\ell_i(v))\}) = I(\hat{S}) - \ell_i(v) = W(t,S) - \ell_i(v) \\
		 W(t,S) &\leq W(t',S\setminus\{(v,\ell_i(v))\}) + \ell_i(v).
	\end{align*}

	On the other hand consider a valid stub set $ \hat{S}' $ for which the maximum is attained in the definition of $ W(t',S\setminus \{(v,\ell_i(v))\}) $. We need to argue that $ \hat{S}' \cup \{(v,\ell_i(v))\} $ is again a valid stub set. First, by assumption that $ S $ is a valid stub set, we immediately get that $ (v,\ell_i(v)) $ does not intersect any $ (u,\ell_u) \in S\setminus \{(v,\ell_i(v))\} = \hat{S}' \cap X_{t'} $. Additionally, by the properties of the nice tree decomposition, we get that $ v $ has no neighbors in $ V_{t'}\setminus X_{t'} $ and with $ \hat{S}' \setminus X_{t'} \subseteq V_{t'}\setminus X_{t'} $ we can conclude that $ \hat{S}' \cup \{(v,\ell_i(v))\}$ is a valid stub set. Furthermore it is considered in the definition of $ W(t,S) $ and we have that
	\begin{align*}
		W(t,S) \geq I(\hat S' \cup \{(v,\ell_i(v))\}) = I(\hat S') + \ell_i(v) = W(t', S \setminus \{(v,\ell_i(v))\}) + \ell_i(v).
	\end{align*}
 	
		\medskip
	\noindent\textit{Forget node.} Suppose $ t $ is a forget node and $ t' $ its only child such that $ X_t = X_{t'} \setminus \{v\} $ for some $ v \in X_{t'} $. Let $ S $ be any subset of $ S(X_t) $, if $ S $ is not a valid stub set we set $ W(t,S) = -\infty $, else
		$W(t,S) = \max \{W(t',S\cup\{(v,\ell_i(v))\}) \mid i = 0,\dots,\delta_v\}$.

	\medskip
	\noindent\textit{Join node.} Suppose that $ t $ is a join node and $ t_1, t_2 $ its two children with $ X_t = X_{t_1} = X_{t_2} $. Again let $ S $ be any subset of $ S(X_t) $. If $ S $ is not a valid stub set we set $ W(t,S) = -\infty $, else
		$W(t,S) = W(t_1,S) + W(t_2,S) - I(S)$.

	\medskip

	It remains to argue about the running time. Let $ m = |V(C)| $. We know there are $ O(\omega m) $ many nodes in the tree $ T $ of the nice tree decomposition~\cite{downey2012parameterized}. For each bag $ t \in T $ we know it has at most $ \omega + 1 $ many elements and each element has at most $ k + 1 $ many possible stubs, hence we have to compute at most $ (k+1)^{\omega + 1} $ values $ W(t,S) $ per node $ t \in T $. At each forget node we additionally need to compute the maximum of up to $ k $ entries. Consequently we perform $ O((k + 1)^{\omega + 2}) $ many operations per node $ t \in T $.	
	All operations can be implemented in $ O(k\omega) $ time. The only problematic one is to test a stub set for validity. 
	We use a modified version of the data structure used in the independent set algorithm by Cygan et al.~\cite{cygan2015parameterized}. See Appendix~\ref{app:tw} for details.
	
	\begin{theorem}\label{thm:treewidth_maxsped}
		Let $ G $ be a simple graph with $m$ edges and $ \Gamma $ a straight-line drawing of $ G $. If the intersection graph $ C $ of $ \Gamma $ has treewidth at most $ \omega \in \mathbb{N} $ and maximum degree $ k \in \mathbb{N} $, problem \maxsped can be solved in $ O(m(k+1)^{\omega + 2} \omega^2 + m \log m) $ time and space.
	\end{theorem}

We remark that Theorem~\ref{thm:tree} shows a better running time in the case of $ C $ being a tree,
	than would follow from Theorem~\ref{thm:treewidth_maxsped} with $\omega=1$.
	Furthermore, since Theorem~\ref{thm:treewidth_maxsped} is exponential only in the treewidth $\omega$ of $C$, it implies that \maxsped is in the class \XP
	\footnote{The class \XP\ contains problems that can be solved in time $O(n^{f(k)})$, where $n$ is the input size, $k$ is a parameter, and $f$ is a computable function.} 
	when parametrized by~$\omega$.

	\subsection{Algorithms for \maxped}
	\label{sec:alg:maxped}
	Let $ C $ be the intersection graph in a \maxped problem. In contrast to \maxsped we need to consider more combinations of stub lengths since they are not necessarily symmetric anymore. In fact there are $ O(k^2) $ possible combinations of left and right stub lengths $ \ell_i(u), \ell_j(u) $ for a vertex $ u \in V(C) $. For the case of $ C $ being a tree our whole argumentation was based solely on the fact that we can subdivide the stub pairs into sets affecting the parent segment or not. This can also be done with the quadratic size sets of all stub pairs and we only get an additional factor of $ k $ in the running time. 
	\begin{corollary}
		\label{thm:maxped_tree}
		Let $ G $ be a simple graph with $m$ edges and $ \Gamma $ a straight-line drawing of $ G $. If the intersection graph $ C $ of $ \Gamma $ is a tree with maximum degree $ k \in \mathbb{N} $, then problem \maxped can be solved in $ O(mk^2 + m \log m) $ time and space.
	\end{corollary}

	In case of $ C $ having bounded treewidth we again did never depend on the symmetry of the stubs, but only on them forming a finite set for each vertex. Consequently we can again just use these quadratic size sets of stub pairs, adding a factor of $ k+1 $ compared to \maxsped, and obtain the following.
	\begin{corollary}
		Let $ G $ be a simple graph with $m$ edges and $ \Gamma $ a straight-line drawing of $ G $. If the intersection graph $ C $ of $ \Gamma $ has treewidth at most $ \omega \in \mathbb{N} $ and maximum degree $ k \in \mathbb{N} $, problem \maxped can be solved in $ O(m (k+1)^{\omega + 3} \omega^2 + m \log m) $ time and space.
	\end{corollary}	

	\section{Experiments}\label{sec:experiments}
	We implemented and tested the tree decomposition based algorithms.\footnote{\url{https://www.ac.tuwien.ac.at/partial-edge-drawing/}}
	To compute the nice tree decomposition we used the ``htd'' library~\cite{DBLP:conf/cpaior/AbseherMW17} version $ 1.2 $,
	compiled with gcc version 8.3.
	The algorithm itself  was implemented in Python3.7, using the libraries\footnote{\url{https://networkx.github.io/} and \url{https://github.com/Toblerity/Shapely}} NetworkX 2.3 and Shapely 1.6.
	To run the experiments we used a cluster, each node equipped with an
	Intel Xeon E5-2640 v4 processors clocked at 2.4GHz, 
	160GB of Ram, 
	and operating Ubuntu 16.04.
	Each run had a memory limit of at most 80GB of RAM.
	
	We generated random graphs using the NetworkX \texttt{gnm} algorithm. 
	The graphs have $n=40$ vertices and between $m=40$ and $75$ edges in increments of 5. This makes a total of 800 graphs, 100 for each $m \in \{40, 45, \dots, 75\}$.
	For the layouts we used the NetworkX implementation of the spring embedder by Fruchterman and Reingold~\cite{DBLP:journals/spe/FruchtermanR91} and the graphviz\footnote{\url{https://www.graphviz.org/}} implementation ``circo'' of a circular layout, version 2.40.1.
	We could successfully run \maxsped for all but four of the spring layouts and for all circle layouts with up to $ 60 $ edges.
	
	\begin{figure}[tbp]
		\includegraphics[trim=0 10 0 0]{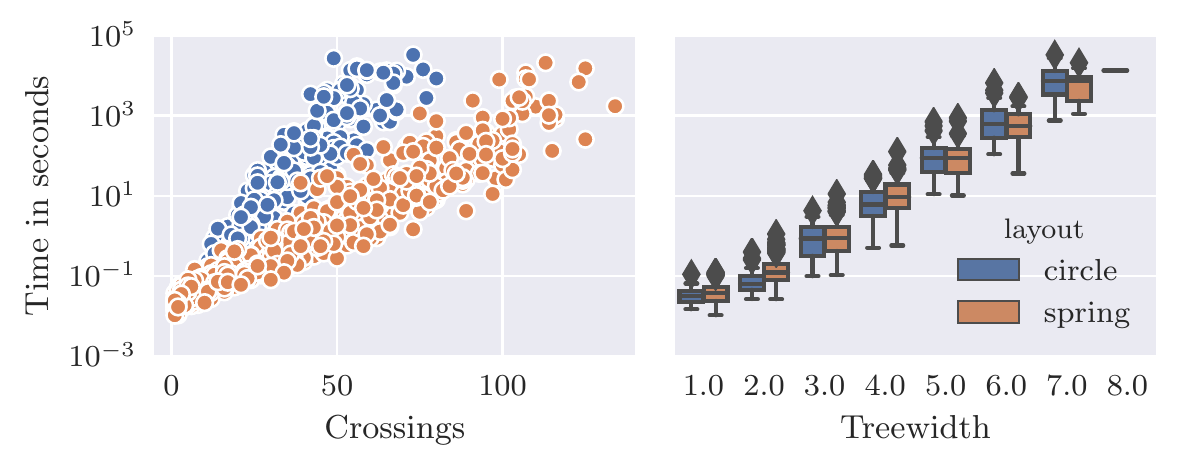}
		\caption{Experimental results for the \maxsped algorithm.}
		\label{fig:maxsped_plot}
	\end{figure}

	Since the time complexity of the algorithm depends exponentially on the treewidth of the intersection graph, we evaluated the running time relative to treewidth and number of crossings, %
	see Fig.~\ref{fig:maxsped_plot}.
	The results are as expected, with the runtime quickly increasing to about 16 minutes (1,000 sec) for 
	80 crossings in case of the spring
	and 50 crossings in case of the circle layouts -- or for a treewidth of $ 6 $ for both layouts.
	On the other hand we can handle up to 50 crossings and a treewidth of 4 for the spring layouts in about 10 seconds.
	The discrepancy in the runtime relative to the number of crossings between spring and circle layouts likely comes from different numbers of crossings per edge.
	To confirm this we took for each intersection graph its maximum degree $k$ divided by the total number of input crossings.
	For the spring layouts this resulted in a ratio of $ 0.24 $ and for the circle layouts $ 0.33 $.
	Recall that the running time is dominated by $ O((k + 1)^{\omega+2}) $. %
	Hence an increase by a of factor 1.5 in the aforementioned value also results in an additional factor of $ 1.5^{\omega+2} $ in the asymptotic running time.
	Concerning ink, for the circle layouts an average of 84\% ($\sigma = 0.09$) and for spring layouts an average of 90\% ($\sigma=0.06$) of the ink could be preserved.
	For \maxped we conducted the same experiments. In general one can say that the additional factor of $ (k+1)^\omega $ makes a big difference. For details see Appendix~\ref{app:ex}.

	In summary, our experiment confirmed the predicted running time behavior and  showed that the amount of removed ink was moderate. Moreover, the ``htd'' library~\cite{DBLP:conf/cpaior/AbseherMW17} performed very well for computing a nice tree decomposition so that we could focus on implementing the dynamic programming algorithm itself.

	\section{Conclusion}
	We extended the work by Bruckdorfer et al.~\cite{bcgkmn-pped-17} and showed \NP-hardness for the \maxped and \maxsped problems, as well as polynomial-time algorithms for the case of the intersection graph of the input drawing being a tree or having bounded treewidth.
	For the latter, our proof-of-concept implementation worked reasonably well for small to medium-size instances.
	
	An interesting open problem is to close the gap for \maxped. While we showed it to be \NP-hard for $ k \geq 4 $ and it is easy to solve for $ k \leq 2 $~\cite{bk-ecbe-12}, the case of $ k = 3 $ remains open. Another direction is to consider the existential question for homogeneous (but non-symmetric) PEDs with a fixed ratio $ \delta $, for which we can freely distribute the $\delta$ fraction of the ink to both stubs. We expect that our algorithms for trees and intersection graphs of bounded treewidth extend to that case, but we could not resolve if the problem is \NP-hard or not.

	\bibliographystyle{splncs04}
	\bibliography{paper}
	
	\newpage
	\appendix
	
	\section{Omitted Proofs from Section~\ref{sec:hardness_maxped}}
	\label{app:hardness}
	Let $ s $ be a line segment in the plane. 
	We define the fractions $ l(s), r(s) \in [0,1], l(s) + r(s) \leq 1 $, such that $ l(s) $ is the fraction of $ s $ drawn from its left endpoint and $ r(s) $ the fraction drawn from its right endpoint. Note that already Bruckdorfer and Kaufmann~\cite{bk-ecbe-12} observed that any 2-plane drawing admits a PED where for every segment $ s $ we find $ l(s) + r(s) = 1 - \varepsilon $ for $ \varepsilon > 0 $. In the following we will omit this $ \varepsilon $ and simply assume it is drawn with its full ink value of $ |s| $, i.e., $ l(s) + r(s) = 1 $. This will simplify the calculations in the proof of Theorem~\ref{thm:maxped_hard} significantly. Still the information where the gap of size $ \varepsilon $ would be placed is preserved by evaluating if $ l(s) > r(s) $ or $ r(s) > l(s) $.
	
	\thmmaxpedhard*
	\begin{proof}

		Let $\phi$ be a planar \textsc{3-Sat} formula with $n$ variables $\{x_1, \dots, x_n\}$ and $m$ clauses $\{c_1, \dots, c_m\}$, each consisting of three literals.
		Analogously to Section~\ref{sec:hardness} we assume that $\phi$ comes with a polynomial-size planar drawing of its variable-clause graph $H_\phi$, which has a vertex for each variable $x_i$ and a vertex for each clause $c_j$.
		Similarly to above we create a 4-plane drawing $\Gamma_\phi$ (given as a set of line segments) such that $\Gamma_\phi$ has a PED with ink at least $L \in \mathbb{N}$ if and only if $\phi$ is satisfiable. 
		
		Let $ p $ be the maximum number of occurrences of any variable $ x $ in $ \phi $. Then for a variable $x$ of $\phi$, we create a gadget $\mathcal{G}(x)$ consisting of a cycle of line segments $\{s_1, s_2, \dots, s_p\}$, where each pair of segments $s_i$ and $s_{i+1} $ intersect as well as the pair $s_p$ and $s_1$. 
		Observe that every $ \mathcal G(x) $ is 2-planar and hence admits a PED with $ I(\mathcal G(x)) = \sum_{i = 0}^p |s_i|$.
		Further, a gadget $\mathcal{G}(x)$ admits exactly two different PEDs of this kind, one where for all segments $s_v$ it holds that $l(s_v) < r(s_v)$ and one where the inverse holds, see  Fig.~\ref{fig:PED_red_vg}. We associate these two configurations with $x$ being set to \emph{true} or \emph{false}, respectively.
		
		\begin{figure}[tbp]
			\centering
			\subfloat[\emph{true} state]{
				\includegraphics{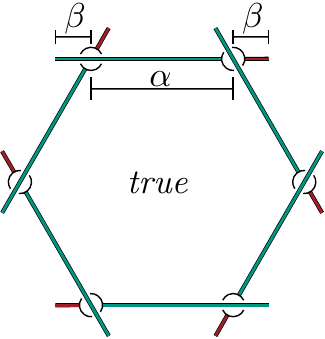}
				\label{fig:PED_red_vg_true}
			}
			\qquad
			\subfloat[\emph{false} state]{
				\includegraphics{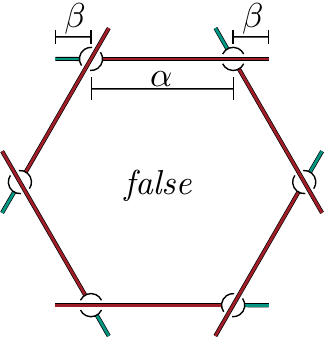}
				\label{fig:PED_red_vg_false}
			}
			\caption{Ink-maximal PED of the variable gadget in both truth states; left stubs are green, right stubs are red.}
			\label{fig:PED_red_vg}
		\end{figure}
		
		It remains to fix the length of each segment $ s_i $ in a variable gadget $ \mathcal G(x) $.  We subdivide such a segment into three parts. The ones left and right of the intersection points with $ s_{i-1} $ and $ s_{i+1} $ and the one in between. Let $ \alpha, \beta \in \mathbb{R}^+ $ with $ \alpha \gg \beta $. The former ones are of length $ \beta $ each and the latter one of length $ \alpha $. Then $ |s_i| = \alpha + 2\beta $ and for one variable gadget we get the maximum ink value $I(\mathcal{G}(x)) = |\mathcal{G}(x)|\cdot(\alpha + \beta)$ in any PED.

		Next we describe the clause gadget. Let $c$ be a clause in $\phi$ and $\mathcal{G}(c)$ its gadget. As mentioned above we use a similar triangular arrangement as in the proof of Theorem~\ref{thm:hard}. We split the gadget into multiple components. The three segments forming the central triangle we name the \emph{internal component} of $\mathcal{G}(c)$, compare Fig.~\ref{fig:PED_red_cg_complete}. In contrast to the proof of Theorem~\ref{thm:hard} we extend every such segment and intersect it with a so called \emph{weight component}. Such a weight component is a cycle made up of four segments. 
		
		More precisely we subdivide one segment $ s_t $ of the internal component into four parts. The two outermost ones are of length $ \beta $ each, the one between the two intersection points with the internal component is of length $ \alpha $, and the fourth one of length $ 3\alpha $, which we will call the \emph{connecting section} of $ s_t $, see Fig.~\ref{fig:PED_red_cg_complete}. For a segment $ s_w $ of the weight component we pick its length left and right of the intersections with the weight component to be of length $ \beta $ and the one enclosed by the two intersection points to be of length $ 18\alpha $. Finally let let $ s_w $ be the segment intersecting with $ s_t $, then their point of intersection is exactly at the midpoint of $s_w$.
		
		In total the clause gadget has $ I_{total}(\mathcal G(c)) =  228\alpha + 30\beta $ units of ink available. As we will argue below the maximum ink value of any PED configuration is at most 
		\begin{align*}
		I(\mathcal{G}(c)) = 3\cdot(72\alpha + 8\beta) + 5\alpha + 6\beta = 221\alpha + 30\beta
		\end{align*}

		\begin{figure}[tbp]
			\centering
			\subfloat[complete clause gadget]{
				\includegraphics{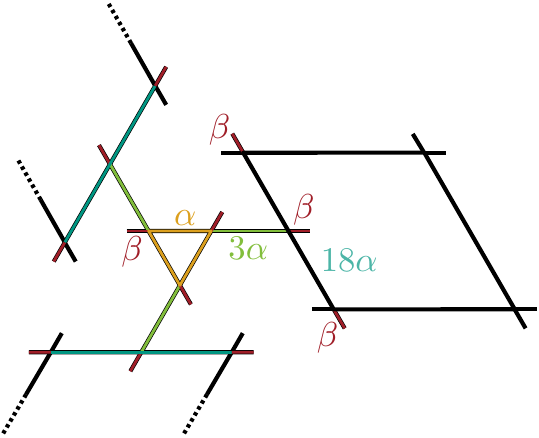}
				\label{fig:PED_red_cg_complete}
			}
			\subfloat[optimal ink-value]{
				\includegraphics{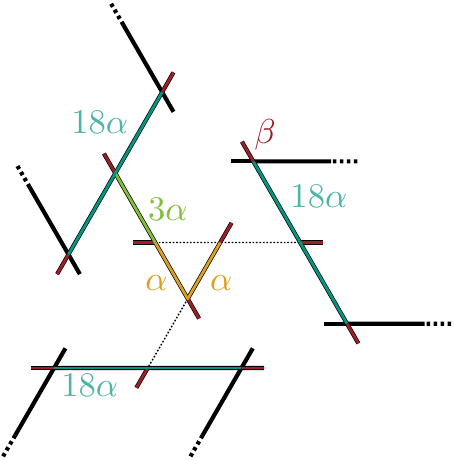}
				\label{fig:PED_red_cg_optimal}
			}
			\caption{A straight-line graph drawing (a) and a maximum-ink SPED (b) of the same graph.}
			\label{fig:PED_red_cg}
		\end{figure}

		Therefore we find $ I_{total}(\mathcal G(c)) - I(\mathcal G(c)) = 7\alpha$. The value $ I(\mathcal G(c)) $ can be obtained only  if we draw the weight components with their maximum ink value, i.e., each segment $ s_t $ of the internal component has to have a gap at the intersection between $ s_t $ and the corresponding segment $ s_w $ of the weight component intersected by $ s_t $, see Fig.~\ref{fig:PED_red_cg_optimal}. If we place the gap of $ s_t $ anywhere else, we lose at least $ 9\alpha $ units of ink in its weight component. This already exceeds the difference of $ 7\alpha $. Since one connecting section has $ 3\alpha $ units of ink, it is never beneficial to draw all three yellow sections of the central triangle in Fig.~\ref{fig:PED_red_cg_optimal}, since this would again lose at least $ 9\alpha $ units of ink.

		The final component of our reduction is the connection between a variable $x$ and a clause $c$ in which $ x $ occurs in the formula $\phi$. This gadget we call a \emph{wire gadget} $\mathcal{G}(x,c)$. One $ \mathcal G(x,c) $ is a set of $z$ line segments $\{s_1, s_2, \dots, s_{z}\}$ for some $ z \in \mathbb{N} $. Each line segment $ s_i \in \mathcal G(x,c), i = 1,\dots,z-1 $  intersects the next line segment $s_{i+1}$. Further, $s_1$ intersects a segment in $\mathcal{G}(x)$ and $s_{z}$ intersects a segment in $\mathcal{G}(c)$.  All segments $ s_i \in \mathcal G(x,c) $ have length $|s| = \alpha/2 + 2\beta$, where again the length left and right of the intersection points is $ \beta $ and the length between the two intersection points is $\alpha/2$.  We will call the latter part the \emph{middle section} of a segment $s_i \in \mathcal G(x,c)$. We call the endpoint of $s_i$ that is next to the intersection with $s_{i+1}$, the \emph{right endpoint} of $s_i$, and the stub connected to this endpoint, the \emph{right stub}.
		
		Let $\mathcal G(x,c)$ be the wire gadget for some variable $ x $ appearing in clause $ c $ in $ \phi $. Then $ s_1 \in \mathcal G(x,c) $ intersects a segment $t_i \in \mathcal{G}(x)$. If $x$ occurs as a positive literal in $c$, the intersection point has a distance of $\beta$ to the intersection point between $t_i$ and $t_{i+1}$. If $x$ occurs as a negative literal in $c$, the intersection point has a distance of $\alpha-\beta$ to the intersection point between $t_i$ and $t_{i+1}$, compare the two cases in Fig~\ref{fig:PED_red_vg_wire}.
		
		Now let $r_j \in \mathcal{G}(c)$ be a segment of the internal component of $\mathcal{G}(c)$. Then $ s_z \in \mathcal G(x,c) $ intersects $r_j$ in the connecting section, at a distance of $\beta$ to the intersection with the next segment of the internal component, see Fig.~\ref{fig:PED_red_cg_wire}.
		Observe that each of the three wire gadgets intersecting a clause gadget is intersecting a different connecting section.
		The maximum ink value of a wire gadget $\mathcal{G}(x, c)$ in a PED is	$I(\mathcal{G}(x, c)) = |\mathcal{G}(x, c)|(\alpha/2 + 2\beta)$.

		\begin{figure}[tbp]
			\centering
			\subfloat[Wires at variable gadget]{
				\includegraphics{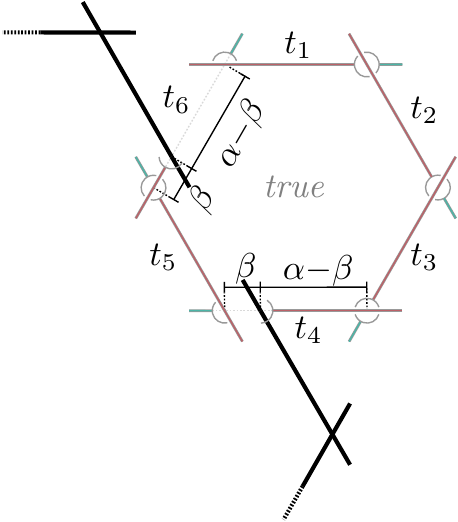}
				\label{fig:PED_red_vg_wire}
			}
			\qquad
			\subfloat[Wires at clause gadget. The red circle marks the gap being placed at the clause.]{
				\includegraphics{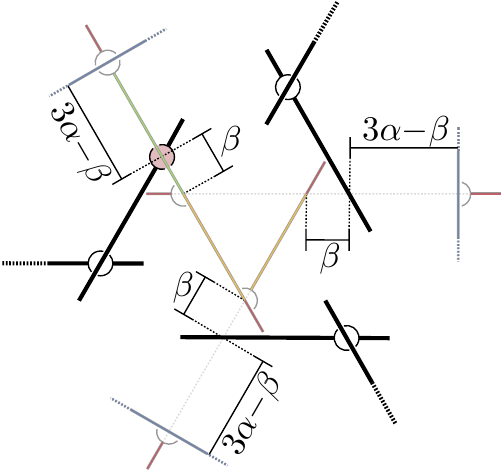}
				\label{fig:PED_red_cg_wire}
			}
			\caption{Wire connections at the variable gadget (a) and the clause gadget (b).}
			\label{fig:PED_red_wires}
		\end{figure}

		We now analyze the possible ways a wire gadget can be drawn in an ink-maximal PED of $ \Gamma_\phi $.
		First let $\mathcal{G}(x,c)$ be a wire gadget for variable $ x $ and clause $ c  $ and assume $ s_z \in \mathcal G(x,c) $ is drawn with a gap at the intersection with $ s \in \mathcal G(c) $ and let the connecting section of $ s $ be drawn. Since we can ever only draw one connecting section of the internal component of $\mathcal{G}(c)$, the two other wire gadgets intersecting segments in $ \mathcal G(c) $ are free to have their last segments drawn completely. Let $\mathcal G(x',c') $ be one of those two wires. We can draw the middle section of every segment $s_i' \in \mathcal G(x',c')$ as part of the right stub. 
		Now $s'_0 \in \mathcal G(x',c')$ has an intersection point with $ r'\in \mathcal G(x') $, splitting $s'_0$ into a left part with length $\beta$ and a right part with length $\alpha/2 + \beta$. Therefore $\mathcal G(x',c')$ has a  PED with $ I(\mathcal G(x',c')) = \sum_{i = 0}^z |s_i|, s_i \in \mathcal G(x',c') $, i.e., all segments can be drawn with full ink. We are left to consider the one wire with a gap at the clause gadget. 

		Let $\mathcal G(x,c)$ be this wire gadget for a variable $ x $ and clause $ c $ in $ \phi $. Then $ s_z \in \mathcal G(x,c) $ has two options at the intersection with $ t \in \mathcal{G}(c)$. Either the middle section of $s_z$ is drawn as part of its right stub, in which case we lose at least $2\alpha$ at $\mathcal{G}(c)$ or it is drawn as part of its left stub, in which case we do not lose any ink, but instead force the middle section of each other segment $s_i$ of $\mathcal G(x,c)$ to be draw as part of its left stub. In particular this means that the middle section of $s_0$ can only be drawn as part of its left stub. This leaves again two options. Either $s_0$ omits its middle section and is only drawn with its two $ \beta $-sized stubs, resulting in a loss of $\alpha/2$ units of ink, or we draw the middle section of $s_0$ as part of its left stub, in which case we either lose $\beta$ or $\alpha - \beta$, depending on the chosen state of $\mathcal{G}(x)$.
		
		We can now compute the exact optimal ink value $L$ for a maximum-ink PED of $\Gamma_\phi$ by adding $I(\mathcal{G}(x))$ for every variable in $\phi$, $I(\mathcal{G}(c))$ for every clause in $\phi$, $I(\mathcal{G}(x,c))$ for every literal in $\phi$ and finally subtracting $\beta$ for every clause in $\phi$.
		
		\begin{align*}
		L = n\cdot I(\mathcal{G}(x)) + m\cdot(I(\mathcal{G}(c) )+ 3\cdot I(\mathcal{G}(x,c)) - \beta)
		\end{align*}
		
		We claim that $\Gamma_\phi$ has a PED with an ink value of at least $L$, if and only if $\phi$ is satisfiable. We will now argue the space requirements and correctness of our reduction now.
		For the space requirements we can use a similar approach as in Theorem~\ref{thm:hard}. 
		We draw all gadgets with their segment endpoints and intersection points on a triangular grid, using $ \beta $ as the unit length of this grid.
		Further we set $ \alpha $ as a suitably large even multiple of $ \beta $.
		Since the original drawing of $ H_\phi $ was of polynomial size this yields a drawing of polynomial size in the number of clauses and variables with integer coordinates in the triangle grid.

		Assume $\phi$ has a satisfying variable assignment. Then every clause $c$ contains at least one  literal set to true. First assume this literal is positive then the corresponding variable $x$ is set to true. We draw $\mathcal{G}(x)$ in its true state (see Fig.~\ref{fig:PED_red_vg}), then we can draw $s_0$ of $\mathcal{G}(x,c)$ with its gap at the intersection with $s_1$ and we only lose $\beta$ units of ink; analogously for the case of a negative literal and $ x $ set to false. In either case we can place the gap of $s_z$ at its intersection with $r_j \in \mathcal{G}(c)$, which allows to draw the connecting section of $ r_j $. As argued above, the other two wire gadgets can be drawn with maximum ink. Since we started with a satisfying assignment, this holds for every $ \mathcal G(c) $ and therefore $\Gamma_\phi$ has a \ped $ D $ with $I(D) \geq L$.

		Now conversely assume $\Gamma_\phi$ has a \ped $ D $ with $ I(D) \geq L$. Construct a variable assignment of $ \phi $ by assigning to each variable $ x $ in $ \phi $ the state derived from the corresponding gadget $ \mathcal G(x) $. Assume for contradiction that this assignment is not satisfying. Then by construction we find three variable gadgets $ \mathcal G(x_1),\mathcal G(x_2)$, and $\mathcal G(x_3) $ connected to one clause gadget $ \mathcal G(c) $, s.t., $ c $ is not satisfied by the assignment derived for $ x_1, x_2, x_3 $. As argued above we know exactly one such connection $ \mathcal G(x_i,c) $ is of the form that $ s_z \in \mathcal G(x_i,c) $ places its gap at the intersection with the corresponding segment in $ \mathcal G(c) $. Consequently we find that $ s_1\in \mathcal G(x_i,c) $ forces the corresponding segment $ t \in \mathcal G(x_i) $ to be drawn without its $ \alpha $ units long section, but this contradicts $ I(D) \geq L $. Alternatively we could omit the middle section for some $ s_j \in \mathcal G(x_i,c), 1 < j < z $, but then again we would lose $ \alpha $ units of ink, which contradicts $ I(D) \geq L$.
\end{proof}

	\section{Omitted Proofs from Section~\ref{sec:btreewidth}}
	\label{app:tw}
	\lemintersection*
	\begin{proof}
		Let $ S_1, S_2 \subseteq S(V(C)) $ be two maximum ink valid stub sets with $ S = S_1 \cap S(X_t) = S_2 \cap S(X_t) $. Assume for contradiction $ I(S_1 \cap S(V_t)) > I(S_2 \cap S(V_t)) $. We can construct a new solution $ S_2' $ by exchanging the stubs in $ S_2 \cap S(V_t \setminus X_t) $ by the stubs in $ S_1 \cap S(V_t\setminus X_t) $. Then $ S_2' $ is again a valid stub set, since $ X_t $ separates the graph induced by $ V_t \setminus X_t $ from the one induced by  $ V(C) \setminus X_t $ while the intersection of $ S_1,S_2 $ with $ X_t $ is the same. By assumption though we find $ I(S_2') > I(S_2) $, a contradiction.
	\end{proof}

	\paragraph*{Forget node.} Suppose $ t $ is a forget node and $ t' $ its only child such that $ X_t = X_{t'} \setminus \{v\} $ for some $ v \in X_{t'} $. Let $ S $ be any subset of $ S(X_t) $, if $ S $ is not a valid stub set we set $ W(t,S) = -\infty $, else
	\begin{align*}
		W(t,S) = \max \{W(t',S\cup\{(v,\ell_i(v))\}) \mid i = 0,\dots,\delta_v\}.
	\end{align*}

	Let $ \hat{S} $ be a valid stub set whose maximum is attained in $ W(t,S) $ and $ i \in [0,\delta_v] $, s.t., $ (v,\ell_i(v)) \in \hat{S} $. Then $ \hat{S} $ is considered in $ W(t', S\cup \{(v,\ell_i(v))\}) $ and hence $ W(t', S\cup \{(v,\ell_i(v))\}) \geq I(\hat{S}) = W(t,S) $ and it follows
	\begin{align*}
		W(t,S) \leq \max \{W(t',S\cup\{(v,\ell_i(v))\}) \mid i = 1,\dots,\delta_v\\ \text{ and } S\cup\{(v,\ell_i(v))\} 	\text{ is a valid stub set}\}.
	\end{align*}

	On the other hand, observe that every stub set in the definition  of $ W(t',S \cup \{(v,\ell_i(v))\}) $ for some $ i \in [0,\delta_v] $ is also considered in the definition of $ W(t,S) $ and hence
	\begin{align*}
		W(t,S) \geq \max \{W(t',S\cup\{(v,\ell_i(v))\}) \mid i = 1,\dots,\delta_v\\ \text{ and } S\cup\{(v,\ell_i(v))\} \text{ is a valid stub set}\}.
	\end{align*}

	\paragraph*{Join node.} Suppose that $ t $ is a join node and $ t_1, t_2 $ its two children with $ X_t = X_{t_1} = X_{t_2} $. Again let $ S $ be any subset of $ S(X_t) $. If $ S $ is not a valid stub set we set $ W(t,S) = -\infty $, else
	\[
	W(t,S) = W(t_1,S) + W(t_2,S) - I(S)
	\]

	Let $ \hat{S} $ be a valid stub set whose maximum is attained in $ W(t,S) $. Let $ \hat{S}_1 = \hat{S} \cap S(V_{t_1}) $ and $ \hat{S}_2 = \hat{S} \cap S(V_{t_2}) $. Since $ \hat{S} $ is a valid stub set we find that $ \hat{S}_1 $ and $ \hat{S}_2 $ are also valid stub sets and further $ \hat{S}_1 \cap S(X_{t_1}) = S $. Hence $ \hat{S}_1 $ is considered in $ W(t_1, S) $ and $ W(t_1, S) \geq I(\hat{S}_1) $. Analogously $ W(t_2, S) \geq I(\hat{S}_2) $ and we find with $ \hat{S}_1 \cap \hat{S_2} = S $ that
	\begin{align*}
		W(t,S) = I(\hat{S}) = I(\hat{S}_1) + I(\hat{S}_2) - I(S) \leq W(t_1,S) + W(t_2, S) - I(S).
	\end{align*}

	Now let $ \hat{S}_1' $ and $ \hat{S}_2' $, be valid stub sets for which the maximum is attained in $ W(t_1,S) $ and $ W(t_2,S) $, respectively. We know that there is no edge between any two vertices in $ V_{t_1}\setminus X_t $ and $ V_{t_2} \setminus X_t $. Consequently no two pairs $ (u,\ell_u) \in \hat{S}_1' $ and $ (v,\ell_v) \in \hat{S}_2' $ intersect and hence $ \hat{S}' := \hat{S}_1' \cup \hat{S}_2' $ is a valid stub set. Moreover $ \hat{S}' \cap S(X_t) = S $ implies that $ \hat{S}' $ is considered in $ W(t,S) $ and we obtain
	\begin{align*}
		W(t,S) \geq I(\hat{S}') = I(\hat{S}_1') + I(\hat{S}_2') - I(S) = W(t_1,S) + W(t_2, S) - I(S).
	\end{align*}

	\paragraph{Implementation details.}Almost all operations can directly be implemented in $ O(k\omega) $ time. The only problematic one is to test a stub set for validity. A naive implementation would check the intersection of $ O(km) $ entries, resulting in a running time quadratic in $ m $. Instead we use a modified version of the data structure used in the independent set algorithm by Cygan et al.~\cite{cygan2015parameterized}. See Appendix~\ref{app:tw} for details.
	We consider for a vertex $ u \in V(C) $ the bag $ X_t, t \in T $ that is closest to the root of $ T $. For any vertex $ v \in V(C) $, $ uv \in E(C) $, and $ v \in X_t $ we add $ v $ to the set $ A_u $ of adjacent vertices of $ u $ stored at $ t $. These sets can be constructed in $ O(\omega m) $ time and space and a query for one adjacency now only takes $ O(\omega) $ time. Here we enrich this structure by adding to each entry in $ A_u $ the index pair $ 1\leq i,j \leq k $ such that the segments corresponding to $ u $ and $ v $ do not intersect if stubs $ \ell_i(u) $ and $ \ell_j(v) $ are chosen. We can then query in $ O(\omega) $ time if two stubs $ (u,\ell_{i'}) $ and $ (v, \ell_{j'}) $ intersect or not. To compute the enriched structure we get an additional factor of $ k $, resulting in a preprocessing time of $ O(k\omega m) $. 
	
	The above description leads to an immediate running time of $ O(m (k+1)^{\omega + 3}\omega^2) $ for the recurrence. 
	We can save a factor of $ k + 1 $ by storing for each node the sets $ S \subseteq S(X_t) $ that are valid stub sets. Then for every node we simply iterate those sets. As a result it is no longer necessary to check for validity in case we are processing a forget or join node. Including the preprocessing time of $ O(m\log m) $ to compute the intersection graph $ C $ we obtain the following result.
	
	\section{Experiments for \maxped}\label{app:ex}
	\begin{figure}[tbh]
		\includegraphics{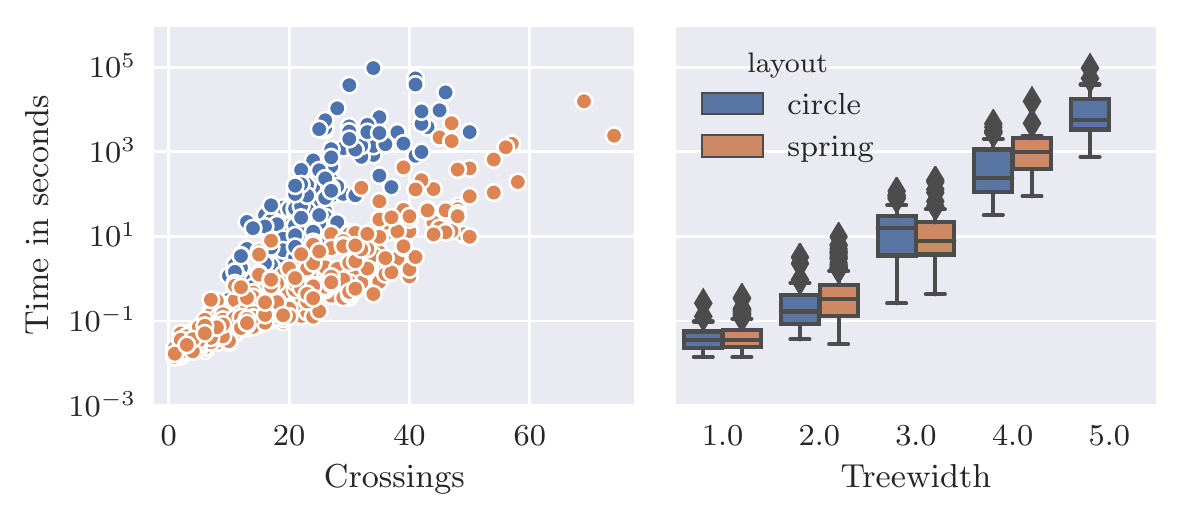}
		\caption{Experimental results for the \maxped algorithm.}
		\label{fig:maxped_plot}
	\end{figure}
	We conducted the same experiment as for \maxsped also for \maxped. 
	We were able to run \maxped for all the spring layouts with $ m = 40, 45, \ldots, 60 $ edges, and for all but seven of the circle layouts with $ m = 40, 45, 50 $ edges. For $ m = 55 $ we were only able to compute the maximum-ink PED for 29 and for $ m = 60 $ only for four of the instances.
	Figure~\ref{fig:maxped_plot} shows the same type of plots as Fig.~\ref{fig:maxsped_plot} for \maxsped. 
	In general one can make similar observations as in the case of \maxsped. 
	It seems that the running time varies more strongly, but given the additional factor of  $ (k+1)^\omega $ in the runtime we also could not compute as many instances as in Section~\ref{sec:experiments}.
	In terms of ink, given the sparsity of the instances it seems to become rarely necessary to draw short edges and almost always the algorithm preserves nearly all ink, $ 99\% $ for spring and $ 97\% $ for the circle layouts.
\end{document}